\documentclass[11pt]{amsart}
\usepackage{amsthm}
\usepackage{amsmath,amstext,amsthm,amsfonts,amssymb,amsthm}
\usepackage{multicol}
\usepackage{latexsym}
\usepackage{color}
\usepackage{graphicx}
\usepackage{epsf}
\usepackage{epsfig}
\usepackage{epic}
\usepackage{graphics}
\usepackage{verbatim}
\usepackage{color}
\usepackage{hyperref}
\newtheorem{theorem}{Theorem}[section]
\newtheorem{lemma}[theorem]{Lemma}
\theoremstyle{definition}
\newtheorem{definition}[theorem]{Definition}
\newtheorem{proposition}[theorem]{Proposition}

\theoremstyle{remark}
\newtheorem{remark}[theorem]{Remark}

\numberwithin{equation}{section}
\newcommand{\abs}[1]{\lvert#1\rvert}

\newcommand{\HI}{\mathfrak{H}}

\newcommand{\R}{\mathbb{R}}

\newcommand{\C}{\mathbb{C}}
\newcommand{\HQ}{\mathbb{H}}

\newcommand{\oz}{\overline{z}}
\newcommand{\qu}{\mathfrak{q}}
\newcommand{\oqu}{\overline{\mathfrak{q}}}

\newcommand{\fn}{\mathfrak{n}}
\newcommand{\Ps}{{\mathbf{P}}^2(\R)}
\newcommand{\Cn}{\C_{\mathfrak{n}}}
\newcommand{\Hr}{\mathfrak H_{\text{reg}}}
\newcommand{\Hrt}{\mathfrak H_{\theta, \phi}^{\text{reg}}}

\begin{document}
\title[Quaternionic coherent states]{Coherent states on quaternion Slices and a measurable field of Hilbert spaces}
\author{B. Muraleetharan$^{\star}$, K. Thirulogasanthar$^{\dagger}$ }
\address{$^{\star}$ Department of Mathematics and Statistics, University of Jaffna, Thirunelveli, Jaffna, Srilanka. }
\address{$^{\dagger}$ Department of Computer science and Software engineering, Concordia University, 1455 de Maisonneuve Blvd. West, Montreal, Quebec, H3G 1M8, Canada.}

\thanks{ KT would like to thank University of Jaffna for their hospitality, part of this work was done while he was visiting there. }
\date{\today}
\dedicatory{Dedicated to the memory of S. Twareque Ali}

\begin{abstract}
A set of reproducing kernel Hilbert spaces are obtained on Hilbert spaces over quaternion slices with the aid of coherent states. It is proved that the so obtained set forms a measurable field of Hilbert spaces and their direct integral appears again as a reproducing kernel Hilbert space for a bigger Hilbert space over the whole quaternions.  Hilbert spaces over quaternion slices are identified as representation spaces for a set of irreducible unitary group representations and their direct integral is shown to be a reducible representation for the Hilbert space over the whole quaternion field. 
\end{abstract}
\maketitle
\pagestyle{myheadings}
\section{Introduction}\label{sec_intro}
The notion of direct integrals was introduced in 1949 by von Neumann in one of his papers in the series {\em On ring of operators}  \cite{VN}, and later, the theory is customarily developed for direct integrals of Hilbert spaces and direct integrals of von Neumann algebras. In fact, direct integrals provides a natural way in the theory of decomposition of von Neumann algebras of operators   and such decompositions play a natural role in mathematical physics problems \cite{Alibk,WW, VW,R}. Direct integral theory was also used by Mackey in his analysis of systems of imprimitivity and in his general theory of induced representations \cite{M}.  Later, following the method of Mackey, in \cite{CT}, quantum mechanics in quaternionic Hilbert spaces was examined along the lines of quaternionic  imprimitivity theorem. In the application point of view, to name a few, based on direct integrals a framework allowing to integrate a parameterized family of reproducing kernels with respect to some measure on the parameter space was developed in \cite{HT} and then it is applied to the so-called Mercer kernels, Kramer sampling, etc. Making use of the direct integral of Hilbert spaces of some nucleus, the approximate dynamical dependence of physical systems on the external parameters was examined in \cite{BG}.\\

In the complex theory, reproducing kernels and Hilbert spaces possessing such kernels are at the core of the theory of coherent states (CS). Whenever we have a family of CS, there is an associated reproducing kernel Hilbert space, and the basis ingredients in constructing such a space are either positive operator valued measures or measurable families of Hilbert spaces. A reproducing kernel, and an associated Hilbert space, can always be defined whenever we are given a positive-definite kernel on a measurable family of Hilbert spaces. The resulting reproducing kernel Hilbert space, while necessarily a space of functions, does not have to be an $L^2$-space. The possibility of embedding it into an $L^2$-type of space requires, in addition, the existence of a resolution of the identity type relation. For a detail argument along these lines one can consult \cite{Alibk}. Recently, in \cite{HP}, using the concept of measurable field of complex Hilbert spaces and their direct integrals it is proved that the formalism of generalized CS leads to a useful characterization of extremal positive operator valued measures.\\

The quaternion field, $\HQ$ can be divided into a family of quaternion slices along a projective plane $\Ps$  and each slice is isomorphic to the complex plane \cite{Gra1}. We consider left quaternion Hilbert space over each slice and thereby obtain a set of Hilbert spaces,
$\displaystyle\mathbb{A}=\left\{V_{\fn}^L~\vert~\mathfrak{n}\in\Ps\right\}$. By following the general procedure outlined in \cite{ThiAli} we obtain CS and thereby reproducing kernel and a reproducing kernel Hilbert space on each element of $\mathbb{A}$. With the aid of the resolution of the identity we also embed each reproducing kernel Hilbert space to a slice-valued $L^2$ space and thereby identify orthonormal dense subset for each reproducing kernel Hilbert space. In doing so we attain a set of Bargmann type reproducing kernel Hilbert spaces,
$\mathfrak{F}=\displaystyle\left\{\HI_{K}^{\fn}~~~\vert~~\fn\in\Ps\right\}$. In \cite{Thi, ThiAli} by applying the same procedure a Bargmann type reproducing kernel Hilbert space, $\HI_K^\HQ$ on a left quaternion Hilbert space $V_\HQ^L$, which is a Hilbert space over the whole quaternion field $\HQ$, was obtained. We shall prove that the set $\mathfrak{F}$ forms a measurable field of Hilbert spaces. Further, the connection between the set $\mathfrak{F}$ and the Hilbert space $\HI_K^\HQ$ appeared to be a direct integral over $\Ps$.
It is well-known that the set of standard complex canonical CS is associated with an irreducible unitary group representation of the Weyl-Heisenberg group. But this is not the case for quaternions,  in \cite{AdMil, Thi} it is proved that an irreducible unitary group representation cannot be allied with the quaternion CS over the representation space $V_\HQ^L$ due to the noncommutativity of quaternions. However, in this article we shall show that the CS labeled by a quaternion slice, considered as vectors in the space $V_{\fn}^L$, can be associated to an irreducible unitary group representation of the Weyl-Heisenberg  type. Further, we shall also demonstrate that the direct integral of the representations so obtained on the slices forms a decomposable operator, and which is a reducible representation for the representation space $V_\HQ^L$.\\

The article is structured as follows. In section 2 we introduce some basic facts about  quaternion slices, a particular fibration of $\R^4$ and identify a measure on it to work with. In section 3 we study some direct integral Hilbert spaces. Section 4 is devoted to the construction of CS, reproducing kernels, reproducing kernel Hilbert spaces and their bases. In section 5 we prove that the set of reproducing kernel Hilbert spaces obtained in section 4 forms a measurable field of Hilbert spaces and also obtain their direct integral. Section 6 deals with group representations over quaternionic Hilbert spaces. Section 7 ends the manuscript with a conclusion.
\section{Mathematical Preliminaries}
\subsection{Quaternions}
Let $\HQ$ denote the field of quaternions. Its elements are of the form $\qu=x_0+x_1i+x_2j+x_3k$ where $x_0,x_1,x_2$ and $x_3$ are real numbers, and $i,j,k$ are imaginary units such that $i^2=j^2=k^2=-1$, $ij=-ji=k$, $jk=-kj=i$ and $ki=-ik=j$. The quaternionic conjugate of $\qu$ is defined to be $\overline{\qu} = x_0 - x_1i - x_2j - x_3k$ and $|\qu|^2=\qu\overline{\qu}=\oqu\qu=x_0^2+x_1^2+x_2^2+x_3^2$ defines a real norm on $\HQ$. 
\subsection{A particular fibration of $\mathbb R^4$ and a coordinatization}\label{sec-coord-fib}
We introduce a particular coordinatization of $\mathbb R^4$ and an associated fibration, which will help us to understand better the geometry of the field of quaternions and the complex planes contained in it. It is well-known that given a quaternion $\mathfrak q$, there exist $x, y \in \mathbb R$, with $y \geq 0$ and  a unit imaginary quaternion $\mathfrak n$ (i.e., with $\mathfrak n^2 = -1$) such that
\begin{equation}
   \mathfrak q = x + \mathfrak n y , \qquad y \geq 0,
\label{q-complex-rep}
\end{equation}
and moreover, the quantities x, y and $\mathfrak n$ are uniquely determined \cite{Col, Gra1}. On the other hand, for a given $\mathfrak n$, the set of all quaternions of the type
$$
   \mathfrak q = x + \mathfrak n y,\qquad x, y \in \mathbb R,
$$
can be identified with a complex plane, which we denote by $\mathbb C_{\mathfrak n}$. Next, the set of all unit imaginary quaternions is identifiable with the  surface of the unit sphere, $S^2$, in $\mathbb R^3$. If we identify antipodal points of this sphere, we get the projective plane $\mathbf{P}^2(\mathbb R)$, which consists of three disjoint sets: the surface of the hemisphere, without the boundary rim, which we denote by $\frac 12 S^2$, a semicircle, without the two endpoints, denoted $\frac 12 S^1$ and a last single point (given by the identification of the endpoints of the semicircle), i.e., $\mathbf{P}^2(\mathbb R) \simeq \frac 12 S^2 \cup \frac 12 S^1 \cup \{\text{point}\}$. We shall continue to write $\mathfrak n$ for points in this projective space. Clearly,
\begin{equation}
   \bigcup_{\mathfrak n \in \mathbf{P}^2(\mathbb R)} \mathbb C_{\mathfrak n} = \mathbb H \; .
\label{quatfib1}
\end{equation}
For any two distinct points $\mathfrak n, \mathfrak n' \in \mathbf{P}^2(\mathbb R)$,
$$ \mathbb C_\mathfrak n \cap \mathbb C_{\mathfrak n'} = \mathbb R .$$

Denote by $\mathbb H^*$ the set of quaternions with non-zero imaginary parts (i.e., the set $\mathbb H$ with the real axis removed) and similarly let $\mathbb C_{\mathfrak n}^*$ be  the set $\mathbb C_\mathfrak n$, with the real axis removed. Then clearly
\begin{equation}
\bigcup_{\mathfrak n \in \mathbf{P}^2(\mathbb R)} \mathbb C_{\mathfrak n}^* = \mathbb H^*.
\label{quatfib2}
\end{equation}
The set $\mathbb H^*$ is open in the usual topology of $\mathbb R^4$ and (\ref{quatfib2}) is a fibration of this space over the base space $\mathbf{P}^2(\mathbb R)$ and having fibres $\mathbb C_{\mathfrak n}^*$. (This fibration may be obtained by defining an equivalence relation on $\mathbb H^*$ as follows: let $\mathfrak q = x + \mathfrak n y, \;
\mathfrak q' = x' + \mathfrak n' y' \in \mathbb H^*$; we say that $\mathfrak q \sim \mathfrak q' \Leftrightarrow \mathfrak n = \pm \mathfrak n'$; it is then clear that $\mathbf{P}^2(\mathbb R) \simeq \mathbb H^*/ \sim$\; .) The affine group $G_{\text{aff}}$, with elements $(b, a), \; b \in \mathbb R, \; a \neq 0$ and composition rule $(b,a)\;(b', a') = (b+ ab', aa')$,  which acts freely and fibre-wise in the manner $(b, a)[x + \mathfrak n y] = (ax + b) + \mathfrak n ay$, leaves each fibre invariant; in fact we may identify each fibre $\mathbb C_{\mathfrak n}^*$ with $G_{\text{aff}}$. Thus, $\mathbb H^*$ is a principal bundle with $G_{\text{aff}}$ as structure group.

 While the above observations are fairly straightforward, it will be useful to see them in terms of a coordinatization of $\mathbb R^4$ (by introducing some sort of {\em cylindrical coordinates}). Let  $\mathbf x = (x_0, x_1, x_2, x_3 ) \in \mathbb R^4$ and let us write, using polar coordinates,
 \begin{eqnarray}
x_0  &=&  r\cos\theta_2 \nonumber \\
x_1  &=&  r\sin\theta_2\sin\theta_1 \cos\phi \nonumber \\
x_2  &=&  r\sin\theta_2\sin\theta_1 \sin\phi \nonumber\\
x_3  &=&  r\sin\theta_2 \cos\theta_1 , \qquad \theta_1\; , \theta_2 \in [0, \pi], \quad \phi \in (0, 2\pi] .
\label{polar-coords}
 \end{eqnarray}
The Lebesgue measure on $\mathbb R^4$ in these coordinates is
\begin{equation}
dx_0\;dx_1\;dx_2\; dx_3 = r^3\; dr\; d\Omega_3 (\theta_1, \theta_2, \phi ), \qquad d\Omega_3 (\theta_1, \theta_2, \phi ) = \sin\theta_1 \sin\theta_2\; d\theta_1\;d\theta_2\;d\phi,
\label{leb-meas}
\end{equation}

Introducing the unit vector,
\begin{equation}
\widehat{\mathbf n}(\theta_1, \phi)   = (\sin\theta_1\cos\phi, \;\sin\theta_1\sin\phi,\; \cos\theta_1 ),
\label{unit-vect}
\end{equation}
the unit vector $\widehat{\mathbf n}_0$ along the $x_0$-axis  and writing $x = r\cos\theta_2, \; y = r\sin\theta_2$, we may write for a point $\mathbf{x}$ in the ``upper half'' plane determined by $\widehat{\mathbf n}_0$ and $\widehat{\mathbf n}(\theta_1, \phi)$,
\begin{equation}
\mathbf x = \widehat{\mathbf n}_0 x + \widehat{\mathbf n}(\theta_1, \phi) y, \qquad x\in \mathbb R, \;\; y \geq 0,
\label{cylind-rep}
\end{equation}
an equation which should be compared to (\ref{q-complex-rep}). The above representation of a point in $\mathbb R^4$ is a sort of cylindrical coordinatization.  The vectors $\widehat{\mathbf n}(\theta_1, \phi)$ for all $\theta_1 \in [0, \pi]$ and $ \phi \in (0, 2\pi]$, constitute all the  points of $S^2$. If we restrict $\theta_1$ to $[0, \frac {\pi}2)$ we get the hemisphere $\frac 12S^2$, without the boundary rim and further, if we let $\theta_2$ run through $[0, 2\pi)$, then $y$ will run through $(-\infty , \infty )$. Thus, we have the alternative coordinatization for a point in $\mathbb R^4$,
\begin{equation}
\mathbf{x} = \widehat{\mathbf n}_0 x + \widehat{\mathbf n}(\theta_1, \phi) y, \qquad \theta_1 \in [0, \frac {\pi}2 ), \;\; \phi \in [0, 2\pi), \;\; x, y \in \mathbb R ,
\label{cylind-rep2}
\end{equation}
For each fixed $\widehat{\mathbf n}(\theta_1, \phi)$, it is now possible to introduce a complex structure on the set of vectors (\ref{cylind-rep2}) (with $x, y \in \mathbb R$),  by identifying $\widehat{\mathbf n}_0$ with a real axis and $\widehat{\mathbf n}(\theta_1, \phi)$ with a complex axis. This reflects the analogous realization of the quaternions $\mathbb H$ in terms of complex planes, as in (\ref{quatfib1}). Thus, from this point of view, the field of quaternions is simply $\mathbb R^4$, equipped with this (fibre-wise) complex structure and the imaginary unit quaternion $\mathfrak n$ is the quaternionic version of the unit vector $\widehat{\mathbf n}(\theta_1, \phi)$. The Lebesgue measure (\ref{leb-meas}) in terms of these variables becomes
\begin{equation}
dx_0\;dx_1\;dx_2\; dx_3 =  \vert y\vert\;\sqrt{x^2 + y^2}\; dx\; dy\; d\Omega (\theta_1,\phi ), \qquad d\Omega (\theta_1,\phi )= \sin\theta_1\;d\theta_1 d\phi .
\label{leb-meas2}
\end{equation}
The points which are left out in this coordinate patch form a set of Lebesgue measure zero.
\begin{remark}\label{Re11}
$\Cn\subset H$ is commutative.
However, elements from two different quaternion slices, $\Cn$ and $\C_{\fn'}$ (for $\fn,\fn'\in\Ps$ with $\fn\not=\fn'$) do not necessarily commute.
\end{remark}
\subsection{Quaternion Hilbert spaces}
For the sake of completeness and to gather the necessary properties, we define left quaternionic Hilbert spaces $V_\HQ^L, V_{\fn}^L$ and the left quaternionic Hilbert space of square integrable functions, $L^2_\HQ(X, \mu)$. For further details we refer the reader to \cite{Ad, Col, ThiAli, Vis}.
\subsubsection{The Hilbert space $V_\HQ^L$:} Let $V_\HQ^L$ be a vector space under left multiplication by quaternion scalars. For $f,g,h\in V_\HQ^L$ and $\qu\in H$, the inner product
$$\langle\cdot\vert\cdot\rangle: V_\HQ^L\times V_\HQ^L\longrightarrow \HQ$$
satisfies the following properties
\begin{enumerate}
\item[(i)]$\langle f\vert g\rangle=\overline{\langle g\vert f\rangle}$
\item[(ii)]$\|f\|^2=\langle f\vert f\rangle>0 ~~\text{unless}~~f=0,~\text{a real norm}$
\item[(iii)]$\langle f\vert g+h\rangle=\langle f\vert g\rangle+\langle f\vert h\rangle$
\item[(iv)]$\langle \qu f\vert g\rangle=\qu \langle f\vert g\rangle$
\item[(v)]$\langle f\vert \qu g\rangle=\langle f\vert g\rangle\overline{\qu}$
\end{enumerate}
We assume that $V_\HQ^L$ together with $\langle\cdot\vert\cdot\rangle$ is a separable Hilbert space.

\subsubsection{The Hilbert space $V_{\fn}^L$:} Let $V_{\fn}^L$ be a vector space under left multiplication by quaternion scalars from $\Cn$. For $f,g,h\in V_{\fn}^L$ and $\qu\in \Cn$, the inner product
$$\langle\cdot\vert\cdot\rangle: V_{\fn}^L\times V_{\fn}^L\longrightarrow \HQ$$
satisfies the following properties
\begin{enumerate}
\item[(i)]$\langle f\vert g\rangle=\overline{\langle g\vert f\rangle}$
\item[(ii)]$\|f\|^2=\langle f\vert f\rangle>0 ~~\text{unless}~~f=0,~\text{a real norm}$
\item[(iii)]$\langle f\vert g+h\rangle=\langle f\vert g\rangle+\langle f\vert h\rangle$
\item[(iv)]$\langle \qu f\vert g\rangle=\qu \langle f\vert g\rangle$
\item[(v)]$\langle f\vert \qu g\rangle=\langle f\vert g\rangle\overline{\qu}$
\end{enumerate}
We assume that $V_{\fn}^L$ together with $\langle\cdot\vert\cdot\rangle$ is a separable Hilbert space.  For any $f,g\in V_{\fn}^L$, the inner product $\langle f\vert g\rangle\in \HQ$. Thereby $\langle f\vert g\rangle$ not necessarily commute with the elements of $\Cn$. This feature alone can differentiate $V_{\fn}^L$ from a complex Hilbert space.

\begin{remark}\label{Re12}As of Remark (\ref{Re11}), the space $\Cn$ is commutative. In this regard, if we define the inner product on $V_{\fn}^L$ as
$$\langle\cdot\vert\cdot\rangle: V_{\fn}^L\times V_{\fn}^L\longrightarrow \Cn,$$
then for any $f,g\in V_{\fn}^L$ the inner product $\langle f\vert g\rangle$ commutes with all the elements of $\Cn$, and thereby the space $V_{\fn}^L$ will behave like a complex Hilbert space.
\end{remark}

\subsubsection{The space $L^2_\HQ(X,\mu)$:} Let $(X, \mu)$ be a measure space and $\HQ$ the field of quaternions, then
$$L_\HQ^2(X, \mu)=\left\{f:X\longrightarrow \HQ~~\vert~~\int_X|f(x)|^2d\mu(x)<\infty\right\}$$
is a left quaternion Hilbert space, with the (left) scalar product
$$\langle f\vert g\rangle=\int_Xf(x)\overline{g(x)}d\mu(x),$$
where $\overline{g(x)}$ is the quaternion conjugate of $g(x)$ and left scalar multiplication $af$, $a\in \HQ$, with $(af)(\qu)=af(\qu)$. (See \cite{Vis} for details).\\
We shall also need the definition of regularity in the sequel.
\begin{definition}(Slice-regular functions \cite{Col})\label{D2}
Let $\Omega\subseteq \HQ$ and a real differentiable (i.e., with respect to $x_i,\; i=0,1,2,3$) function $f:\Omega\longrightarrow \HQ$ is said to be slice left regular if, for every quaternion $\fn\in\Ps$, the restriction of $f$ to $\Cn$, has continuous partial derivatives (with respect to $x$ and $y$, every element in $\Cn$ is being uniquely expressible as $x + \fn y$) and satisfies
\begin{equation}
\overline{\partial}_\fn f (x + \fn y) := \frac 12\left(\frac {\partial f_\fn (x + \fn y )}{\partial x}
       + \fn \frac {\partial f_\fn (x + \fn y )}{\partial y}\right) = 0\; .
\label{leftslicereg}
\end{equation}
Similarly,  it is said to be slice right regular if
\begin{equation}
\overline{\partial}_\fn f (x + \fn y) := \frac 12\left(\frac {\partial f_\fn (x + \fn y )}{\partial x}
       +  \frac {\partial f_\fn (x + \fn y )}{\partial y}\fn\right) = 0\; .\\
\label{rightslicereg}
\end{equation}
The anti-regularity can be defined in a similar way.
\end{definition}
\begin{remark}\label{RRem1}
In view of the above definition, in partcular, a function $f:\HQ\longrightarrow \HQ$ is said to be regular or anti-regular if, for each $\fn\in \Ps$ the restriction of $f$ to $\Cn$, $f\vert_{\Cn}$ is regular or anti-regular respectively.
\end{remark}
We shall also need the following theorem in the sequel. Following the proof, line by line, of complex Hilbert spaces (for a proof in the complex case see \cite{Vag}, page 63) it can be proved for the quaternion Hilbert spaces as well.
\begin{theorem}\label{T6}
For an orthonormal sequence $\{u_n\}$ in a left quaternion Hilbert space $\mathfrak{H}$
\begin{enumerate}
\item[(i)]if $f\in \mathfrak{H}$, $\langle u_n\vert f\rangle=0, \forall n \Leftrightarrow f=0$, then $\{u_n\}$ is said to be dense or complete in $\mathfrak{H}$.\\
Other equivalent ways of charecterizing completenes for orthonormal sets are:
\item[(ii)]For every $f\in\mathfrak{H}$, $\sum_{n=1}^{\infty}\abs{\langle u_n\vert f\rangle}^2=\|f\|^2$, the Parseval's equation.
\item[(iii)] For all $f,g\in\mathfrak{H}$, the inner product $\langle f\vert g\rangle$ satisfies
$$\langle f\vert g\rangle=\sum_{n=1}^{\infty}\langle f\vert u_n\rangle\langle u_n\vert g\rangle.$$
\end{enumerate}
\end{theorem}

\section{Some direct integral Hilbert spaces}\label{dir-int-sp}
Using the above fibration and  splitting of the Lebesgue measure, we have the direct integral representation of the complex Hilbert space,
\begin{equation}
L^2_{\mathbb C} (\mathbb R^4,\; dx_0\;dx_1\;dx_2\; dx_3) \simeq  \int_{\theta \in [0, \frac {\pi}2)} \! \int_{\phi \in [0, 2\pi)}^\oplus \mathfrak H_{\theta, \phi}\; d\Omega (\theta,\phi ),
\label{decomp1}
\end{equation}
where, for each $(\theta, \phi), \; \mathfrak H_{\theta, \phi}$ is a copy of the Hilbert space
$L^2_{\mathbb C} (\mathbb R^2, \; \vert y\vert\sqrt{x^2+ y^2}\; dx\;dy)$, which is essentially a Hilbert space of functions on the fibre determined by $\widehat{\mathbf n}(\theta, \phi)$. Here, it ought to be noted that the real axis is a set of measure zero (for the measure $\vert y\vert\sqrt{x^2+ y^2}\; dx\;dy$). If we equip $\mathbb R^2$ with a complex structure, we may also write,
\begin{equation}
\mathfrak H_{\theta, \phi} = L^2_{\mathbb C} (\mathbb C, \; \vert z\;Im\{ z\}\vert\frac {dz\wedge d\overline{z}}{2i} ), \qquad z = x +iy,
\label{comp-rep}
\end{equation}
It is now clear, that if we equip $\mathbb R^4$, with the structure of quaternions, and consider the Hilbert space $L^2_{\mathbb C} (\mathbb H, \; dx_0\;dx_1\;dx_2\; dx_3)$, then it too has exactly the same direct integral decomposition as (\ref{decomp1}).

Going a step further, it is not hard to see that if we consider the (left, or right) quaternionic Hilbert space $L^2_{\mathbb H} (\mathbb H, \; dx_0\;dx_1\;dx_2\; dx_3)$, of functions $f: \mathbb H \longrightarrow \mathbb H$, then it has a similar direct integral decomposition,
\begin{equation}
L^2_{\mathbb H} (\mathbb H,\; dx_0\;dx_1\;dx_2\; dx_3) \simeq \int_{\theta \in [0, \frac {\pi}2)} \! \int_{\phi \in [0, 2\pi)}^\oplus \mathfrak H_{\theta, \phi}\; d\Omega (\theta,\phi ),
\label{decomp2}
\end{equation}
where now $\mathfrak H_{\theta, \phi}$ is a copy of the (left or right) quaternionic Hilbert space $L^2_{\mathbb H} (\mathbb C, \; \vert z\;Im z\vert\frac {dz\wedge d\overline{z}}{2i} )$ of functions $f: \mathbb C \longrightarrow \mathbb H$, which again may be considered to be a Hilbert space of quaternion valued functions on the fibre determined by $\widehat{\mathbf n}(\theta, \phi)$.

In the next section we look at other direct integral decompositions of quaternionic Hilbert spaces and see how they are naturally related to families of coherent states.
\section{Some families of quaternionic nonlinear coherent states}\label{quat-nc-cs}
Suppose that we equip each complex plane $\mathbb C_{\mathfrak n}$ in (\ref{quatfib1}) with a measure $d\nu (z, \overline{z})$, for which the real axis constitutes a set of zero measure. For each $\widehat{\mathbf n}(\theta, \phi)$ (corresponding to the imaginary  unit quaternion $\mathfrak n$), we define the quaternionic Hilbert space $\mathfrak H_{\theta, \phi} = L^2_{\mathbb H} (\mathbb C, \; d\nu (z, \overline{z}) )$. Then, we again have a direct integral decomposition
\begin{equation}
L^2_{\mathbb H} (\mathbb H,\; d\nu (z, \overline{z})\;d\Omega (\theta,\phi ) ) \simeq \int_{\theta \in [0, \frac {\pi}2)} \! \int_{\phi \in [0, 2\pi)}^\oplus \mathfrak H_{\theta, \phi}\; d\Omega (\theta,\phi ).
\label{decomp5}
\end{equation}
In order to build an interesting family of quaternionic coherent states, we shall look at cases where the Hilbert space $L^2_{\mathbb H} (\mathbb H,\; d\nu (z, \overline{z})\;d\Omega (\theta,\phi ) )$ contains subspaces of (right or left) regular functions (of a specific type).

Writing $z = re^{i\vartheta}$, let the measure  $d\nu$ be of the form
\begin{equation}
   d\nu (z, \overline{z}) = d\mu (r^2)\; d\vartheta, \qquad r \in \mathbb R^+, \;\; \vartheta \in [0, 2\pi) ,
\label{moment-meas}
\end{equation}
where the measure $d\mu$ is assumed to have moments of all orders:
\begin{equation}
\mu_n = \int_0^\infty x^n\; d\mu (x) < \infty, \quad n =0,1,2, ... \infty.
\label{moments}
\end{equation}
We normalize this measure so that $\mu_0 = 1$. Next, defining the sequence of positive numbers,
\begin{equation}
x_n = \frac {\mu_n}{\mu_{n-1}}, \qquad n = 1,2, 3, \ldots , \quad x_0 \equiv 1,
\end{equation}
we assume that the series
\begin{equation}
\mathcal N (r) := \sum_{n=0}^\infty \frac {r^{2n}}{x_n!} < \infty, \qquad x_n! = x_1 x_2 x_3 \ldots x_n = \mu_n,
\end{equation}
converges for $0< r < l$, where $l$ could be finite or infinite. In the physical literature, in order to ensure the  self-adjointness of some associated operators it is assumed that the sum $\sum_{n=0}^\infty \frac 1{\sqrt{x_n}}$ diverges and $d\mu$ has support in  $(0, l)$ \cite{Akh}.

Assume now that $\mathfrak H = L^2_{\mathbb H} (\mathbb H,\; d\nu (z, \overline{z})\;d\Omega (\theta,\phi ) )$ is a left quaternionic Hilbert space, with scalar product
$$
  \left(\Phi \mid \Psi \right)_{\mathfrak H} = \int_{\mathbb H} \Phi (\mathfrak q )\overline{\Psi (\mathfrak q )}\; d\nu (z, \overline{z})\; d\Omega (\theta, \phi ). $$
In this space the vectors,
\begin{equation}
\Phi_n (\mathfrak q ) = \frac {\mathfrak q^n}{2\pi\sqrt{x_n !}}, \qquad n =0,1,2, \ldots , \infty,
\label{on-set}
\end{equation}
form an orthonormal set:
$$ \left(\Phi_m \mid \Phi_n \right)_{\mathfrak H} = \delta_{mn}, \qquad m,n = 0,1,2, \ldots , \infty .$$
The restrictions of these functions to $\mathbb C_{\mathfrak n}$:
$$ \phi_n ( z) :=\Phi_n (\mathfrak q )\vert_{\mathfrak n} = \frac {z^n}{2\pi\sqrt{ x_n !}}, \qquad z = x + \mathfrak n y, $$
form an orthogonal set in $\mathfrak H_{\theta,\phi} = L^2_{\mathbb H} (\mathbb C,\; d\nu (z, \overline{z})) )$:
$$
\langle \phi_m \mid \phi_n\rangle_{\theta, \phi} =  \int_{\mathbb C} \phi_m (z )\overline{\phi_n (z)}\; d\nu (z, \overline{z}) = \frac 1{2\pi}\; \delta_{mn}.$$
The vectors $\Phi_n (\mathfrak q )$ span a subspace of $\mathfrak H$, consisting of (right slice-) regular functions, which we denote by $\mathfrak H_{\text{reg}}$. The fact that this is a proper subspace of regular functions can be proved by using the decomposition of a regular function into two holomorphic functions on each slice $\mathbb C_{\mathfrak n}$ \cite{ThiAli}. Similarly, the restricted functions, $\phi_n (z)$, on each slice $\mathbb C_{\mathfrak n}$ generate a (proper) subspace of $\mathfrak H_{\theta, \phi}$, consisiting of analytic functions in the variable $z = x + \mathfrak n y$. Denoting this subspace by $\mathfrak H_{\theta, \phi}^{\text{reg}}$, we have the decomposition
\begin{equation}
\mathfrak H_{\text{reg}} \simeq \int_{\theta \in [0, \frac {\pi}2)} \! \int_{\phi \in [0, 2\pi)}^\oplus \mathfrak H_{\theta, \phi}^{\text{reg}}\; d\Omega (\theta,\phi ).
\label{dir-decomp7}
\end{equation}

Let $\{f_n\}_{n=0}^{\infty}$ be an orthonormal basis of $V_{\HQ}^L$. We define a family of {\em quaternionic nonlinear coherent states\/,} $\eta_{\overline{\mathfrak q}} \in V_{\HQ}^L, \; \mathfrak q \in \mathbb H$, as
 \begin{equation}
 \eta_{\overline{\mathfrak q}} = \frac 1{\sqrt{\mathcal N (r)}} \sum_{n=0}^\infty  \overline{\Phi_n (\mathfrak q )}f_n = \frac 1{2\pi\sqrt{\mathcal N (r)}}\sum_{n=0}^\infty  \frac {\overline{\mathfrak q}^n}{\sqrt{x_n !}}f_n\; .
 \label{quat-cs1}
 \end{equation}
 It is easy to check that these vectors are normalized, $\Vert \eta_{\overline{\mathfrak q}} \Vert^2 = 1$ and satisfy a {\em resolution of the identity\/,} in the sense that for any two vectors $\Psi_1, \Psi_2 \in V_{\HQ}^L$, the relation
 $$
   \int_{\mathbb H}\left(\Psi_1 \mid  \eta_{\overline{\mathfrak q}}\right) \left(  \eta_{\overline{\mathfrak q}}\mid \Psi_2\right)\;\mathcal N (r)\; d\nu (z, \overline{z})\;  d\Omega (\theta , \phi ) = \left(\Psi_1 \mid \Psi_2\right) \;  $$
holds, which we then write as the operator integral:
\begin{equation}
\int_{\mathbb H}\mid  \left. \eta_{\overline{\mathfrak q}}\right) \left(  \eta_{\overline{\mathfrak q}}\mid \right.\;\mathcal N (r)\; d\nu (z, \overline{z})\;  d\Omega (\theta , \phi ) = I_{V_{\HQ}^L} \; .
\label{quat-resolid}
\end{equation}

If we restrict the functions $\Phi_n (\mathfrak q)$ to the slice $\mathbb C_{\mathfrak n}$, we get the vectors in $V_{\fn}^L$
 \begin{equation}
\eta_{\overline{z}} = \frac 1{2\pi\sqrt{\mathcal N (r)}}\sum_{n=0}^\infty \frac {\overline{ z}^n}{\sqrt{x_n !}}g_n  \;, \qquad z = x + \mathfrak n y \in \mathbb C_{\mathfrak n}, \;\; r = \sqrt{x^2 + y^2}\; ,
\label{quat-cs2}
\end{equation}
where $\{g_n\}$ is an orthonormal basis of $V_{\fn}^L$.
They have the normalization $2\pi \Vert \eta_{\overline{z}}\Vert^2_{V_{\fn}^L} = 1$ and satisfy the resolution of the identity
\begin{equation}
2\pi \int_{\mathbb C_{\mathfrak n}}\mid  \left. \eta_{\overline{z}}\right) \left(  \eta_{\overline{z}}\mid \right.\;\mathcal N (r)\; d\nu (z, \overline{z}) = I_{V_{\fn}^L} \; .
\label{quat-resolid2}
\end{equation}
The coherent states in (\ref{quat-cs2}) look exactly like the nonlinear coherent states of quantum physics, except that they are now elements of a quaternionic Hilbert space.
\begin{remark}
One can also consider $V_{\HQ}^L=\Hr$ then the orthonormal basis $\{f_n\}$ in the CS (\ref{quat-cs1}) should be $\{\Phi_n\}$. In this case, the space $V_{\fn}^L=\Hrt$ and $g_n$ in the CS (\ref{quat-cs2}) will be the restricted vector $g_n=\Phi_n\vert_\fn=\phi_n$.
\end{remark}
\subsection{Reproducing Kernels} 
From the general construction (see the appendix), the map
\begin{equation}\label{ee3}
W:V_{\HQ}^L\longrightarrow L^2_{\HQ}(\HQ, d\nu(z,\oz)d\Omega(\theta,\phi))~~~~\text{with}~~~~Wf(\qu)=2\pi\mathcal{N}(r)^{\frac{1}{2}}\langle f\vert \eta_{\oqu}\rangle
\end{equation}
is a linear isometry onto a closed subspace
$$\mathfrak{H}_K^{\HQ}=WV_{\HQ}^L\subset L^2_{\HQ}(\HQ, d\nu(z,\oz)d\Omega(\theta,\phi))$$
and the space $\mathfrak{H}^{\HQ}_K$ is a reproducing kernel Hilbert space with reproducing kernel
\begin{equation}\label{ee4}
K:\HQ\times \HQ\longrightarrow \HQ,\quad K(\qu_1,\overline{\qu}_2)=\sum_{m=0}^{\infty}\frac{\oqu_1^m{\qu}_2^m}{4\pi^2x_m!}
\end{equation}
and the kernel satisfies the following properties.
\begin{enumerate}
	\item [(a)] hermiticity,~$K(\qu_{1},\overline{\qu_{2}})=\overline{K(\qu_{2},\overline{\qu_{1}})}$~~ for all~~$\qu_{1},\qu_{2}\in \HQ;$\\
	\item [(b)] positivity, ~$K(\qu,\overline{\qu})\geq 0$~~ for all~~$\qu\in \HQ;$\\
	\item [(c)] idempotence, $\displaystyle \int_{\HQ}K(\qu_{1},\overline{\qu_{2}})K(\qu_{2},\overline{\qu_{3}})d\nu(z,\oz)d\Omega(\theta,\phi)=K(\qu_{1},\overline{\qu_{3}})$~~ for all~~$\qu_{1},\qu_{3}\in \HQ$.
\end{enumerate}	
Further
\begin{equation}\label{ee6}
\mathfrak{H}_K^{\HQ}=\overline{\text{left span over}~ \HQ}\left\{\Phi_m(\qu)=\frac{{\qu}^m}{2\pi\sqrt{ x_m!}}~~\vert~~m\in\mathbb{N}\right\}
\end{equation}
is a space of right regular functions. Further, from equation (\ref{ee6}), the set
\begin{equation}\label{ee6-a}
B_{\HQ}=\left\{\frac{{\qu}^m}{2\pi\sqrt{x_m!}}~~\vert~~\qu\in \HQ,~~m\in\mathbb{N}\right\}
\end{equation}
 is total in $\mathfrak{H}_K^{\HQ}$.

\subsection{Reproducing kernels on $V_{\fn}^L$}
Now let us define the sequence of functions
\begin{equation}\label{ee7}
U_{m}:\C_{\fn}\longrightarrow \HQ;~~~\text{by}~~~ U_{m}(z)=\frac{z^{m}}{\sqrt{2\pi x_m!}}~~~\text{for all}~~  m\in\mathbb{N}.
\end{equation}
The functions $U_m$ satisfy
\begin{enumerate}
\item[1.] $\displaystyle U_m(z)\in L_{\HQ}^2(\C_\fn, d\nu(z,\oz))\quad\text{for all}~~m\in\mathbb{N}~~\text{and for all}~~z\in \C_\fn$
\item[2.] $\displaystyle 0<\sum_{m=0}^{\infty}|U_m(z)|^2<\infty\quad\text{for all}~~z\in \C_\fn$
\item[3.]$\displaystyle\int_{\C_\fn}U_m(z)\overline{U_n(z)}d\nu(z,\oz)=\delta_{mn}$.
\end{enumerate}

Therefore, from the general construction (see the appendix), the function
$$K:\C_\fn\times \C_\fn\longrightarrow \HQ~~~~\text{with}$$
\begin{equation}\label{eq7}
K(z_{1},\overline{z_{2}})=
\sum_{m=0}^{\infty}\frac{\oz_{1}^{m}{z_{2}}^{m}}{2\pi x_m!},~~~\mbox{~~for all~~}~z_{1},z_{2}\in \C_\fn,
\end{equation}
is a reproducing kernel.
Define the function $W_{\fn}:V^{L}_{\fn}\longrightarrow L^2_\HQ (\C_\fn, d\nu(z,\oz ))$ with
\begin{equation}\label{eq10}
 W_\fn f (z) =2\pi \mathcal N (r)^{\frac 12}\langle f\mid \eta_{\oz}\rangle_{V^L_\fn}
\mbox{~~~for all~~} z\in \C_\fn \mbox{~~~and~~}f\in V^{L}_\fn
\end{equation}
and
$$ \mathfrak H_{K}^{\fn} := W_{\fn}V^{L}_{\fn} \subset  L^2_\HQ (\C_\fn,  d\nu(z,\oz) ).$$
By contruction $W_{\fn}$ is a linear isometry onto the closed subspace $\mathfrak H_{K}^{\fn}$ and it is a reproducing kernel Hilbert space. Further
\begin{equation}\label{eq16}
\mathfrak H_{K}^{\fn} =\overline{\text{left span\,over}\HQ}\{\frac{z^{m}}{\sqrt{2\pi x_m!}}\vert z\in \C_\fn\mbox{~~and~~}m\in\mathbb{N}  \}
\end{equation}
That is,
\begin{equation}\label{ee6-b}
B_{\fn}=\left\{\frac{z^m}{\sqrt{2\pi x_m!}}~~\vert~~z\in \C_\fn,~~m\in\mathbb{N}\right\}
\end{equation}
is total in $\mathfrak{H}_K^\fn$.

\begin{remark}\label{Re6}We shall need the following facts about $\Ps$ in the following section. $\Ps$ is a locally compact set with  the integral measure 
$d\Omega(\theta,\phi)=\sin\theta d\phi d\theta.$ Further
\begin{equation}\label{eq8}
\fn= \cos\theta\sin\phi i+ \sin\theta\sin\phi j+ \cos\phi k~;
\end{equation}
also
\begin{equation}\label{I1}
\int_0^{\frac{\pi}{2}}\int_0^{2\pi} \fn d\Omega(\theta,\phi)=0
\end{equation}
and
\begin{equation}\label{I2}
\int_0^{\frac{\pi}{2}}\int_0^{2\pi}d\Omega(\theta,\phi)  =2\pi.
\end{equation}
\noindent
\end{remark}
\section{Measurable field and Direct integral}
In this section we prove that the bundle of reproducing kernel Hilbert spaces obtained in the above section forms a measurable field of Hilbert spaces and their direct integral is isomorphic to the reproducing kernel Hilbert space $\HI_K^{\HQ}$.  In order to do so we acquire the following background materials from \cite{Alibk, Ta}.
\subsection{Background Theory}
Let $Y$ be any locally compact space, equipped with a Borel measure. Suppose that, for each $y\in Y$, we associate a Hilbert space $\mathfrak K_{y}$. Let $\left\langle \cdot \mid \cdot\right\rangle_{y}$ and $\Vert\cdot\Vert_{y}$ denote the inner product and norm, respectively, in $\mathfrak K_{y}$. Here, we have to make an assumption that the Cartesian product $\displaystyle \prod_{y\in Y}\mathfrak K_{y}$ has a natural vector space structure.
\begin{definition}\label{d41}\cite{Alibk}
The family $\{\mathfrak K_{y}~\mid~y\in Y\}$ is called a measurable field of Hilbert spaces, if there exists a subspace $\mathfrak M$ of the product space $\displaystyle \prod_{y\in Y}\mathfrak K_{y}$ such that,
\begin{enumerate}
	\item for each $\Phi\in\mathfrak M$, the positive, real-valued function $y\longmapsto\Vert\Phi(y)\Vert_{y}$ on $Y$ is $\nu$-measurable;
	\item if for any $\displaystyle\Phi\in \prod_{y\in Y}\mathfrak{K}_{y}$, the complex-valued functions $y\longmapsto\left\langle \Phi(y)\mid \Psi(y)\right\rangle_{y}$, for all $\Psi\in\mathfrak{M}$, are $\nu$-measurable then $\Phi\in\mathfrak{M}$~; and
	\item there exists a countable subset $\{\Phi_{n}\}_{n=1}^{\infty}$ of $\mathfrak{M}$ such that for each $y\in Y$ the set of vectors $\{\Phi_{n}(y)\}_{n=1}^{\infty}$ is total in $\mathfrak{K}_{y}$.
\end{enumerate}
\end{definition}
Elements in $\mathfrak{K}_{y}$ are called $\nu$-{\em{measurable vector fields}} and the sequence $\{\Phi_{n}\}_{n=1}^{\infty}$ a {\em{fundamental sequence}} of $\nu$-measurable vector fields. Measurable field of Hilbert spaces are convenient to construct the {\em{direct integral}} of Hilbert spaces. Next lemma helps us to identify the measurable fields of Hilbert spaces.
\begin{lemma}\label{141}\cite{Alibk}
Let $\{\Phi_{n}\}_{n=1}^{\infty}\subseteq\displaystyle\prod_{y\in Y}\mathfrak{K}_{y}$ satisfy,
\begin{enumerate}
	\item for each $m$ and $n$, the function $y\longmapsto\left\langle \Phi_{m}(y)\mid \Phi_{n}(y)\right\rangle_{y}$ on $Y$ is $\nu$-measurable\; ; and
	\item for each $y\in Y$, the sequence of vectors $\{\Phi_{n}(y)\}_{n=1}^{\infty}$ is total in $\mathfrak{K}_{y}$.
\end{enumerate}
Then the set
$$\mathfrak M=\{\Psi\in\prod_{y\in Y}\mathfrak{K}_{y}~\mid~y\longmapsto\left\langle \Phi(y)_{n}\mid \Psi(y)\right\rangle_{y}~~\mbox{~~is~~}~\nu-\mbox{measurable for all~~}~n\}$$
satisfies the conditions 1., 2. and 3. of Definition (\ref{d41}), and hence $\{\mathfrak K_{y}~\mid~y\in Y\}$ is a measurable field of Hilbert space and $\mathfrak M$ a $\nu$-measurable field of vectors.
\end{lemma}
\begin{definition}\cite{Alibk}
Suppose that we are give $\{Y,\nu\}$ as before and a measurable field of Hilbert spaces $\{\mathfrak K_{y}~\mid~y\in Y\}$ along with  the set of $\nu$-measurable vector field $\mathfrak M$. Let $\overline{\mathfrak H}\subseteq\mathfrak M$ be the collection of all ($\nu$-equivalence class of) vector field $\Phi$ satisfying
\begin{equation}\label{eq27}
 \Vert\Phi\Vert^{2}:=\int_{Y}\Vert\Phi(y)\Vert_{y}^{2}d\nu(y)<\infty,
\end{equation}
and define on it the scalar product
\begin{equation}\label{eq28}
 \left\langle \Phi\mid\Psi\right\rangle:=\int_{Y}\left\langle \Phi(y)\mid\Psi(y)\right\rangle_{y}d\nu(y)\; ,~~\Phi,\Psi\in\overline{\mathfrak H}.
\end{equation}
It can be shown that $\overline{\mathfrak H}$ is complete in the norm (\ref{eq27}) and hence , equipped with the scalar product (\ref{eq28}), it becomes a Hilbert space. We call $\overline{\mathfrak H}$ the {\em{direct integral}} of Hilbert spaces $\{\mathfrak K_{y}~\mid~y\in Y\}$ and write
$$\overline{\mathfrak H}=\int_{Y}^{\oplus}\mathfrak K_{y}d\nu(y).$$
\end{definition}
{\em{Note:}}
 The definition of the product space is
$$\prod_{y\in Y}\mathfrak{K}_{y}:=\{f:Y\longrightarrow\bigcup_{y\in Y}\mathfrak K_{y} ~\mid~~ f(y)\in\mathfrak K_{y},~~\mbox{~~for all~~} y\in Y \}.$$
	
\subsection{Measurable field of reproducing kernel Hilbert spaces}
We  know that $\Ps$ is locally compact and  $d\Omega(\theta,\phi)=\sin\theta d\phi d\theta$ is a measure on it. For each $\fn\in\Ps,~~\mathfrak H_{K}^{\fn}$ is a Hilbert space. So in the following we adapt $Y=\Ps$. Let
\begin{equation}\label{m1}
\overline{\mathfrak{H}}_K=
\left\{\left(h\vert_{\C_\fn}\right)_{\fn\in\Ps}
=\left(h\vert_{\fn}\right)_{\fn\in\Ps}~~\vert~~h\in\mathfrak{H}_K^\HQ\right\}
\end{equation}
In the following proposition we prove that $\overline{\mathfrak{H}}_K$ is isomorphic to $\mathfrak{H}_K^\HQ$.
\begin{proposition}\label{piso}
The Hilbert space $\mathfrak{H}_K^\HQ$ is isomorphic to $\overline{\mathfrak{H}}_K$. i.e. $\mathfrak{H}_K^\HQ\cong\overline{\mathfrak{H}}_K$.
\end{proposition}
\begin{proof}
Define $$\varphi:\mathfrak{H}_K^\HQ\longrightarrow\overline{\mathfrak{H}}_K$$
by $$\varphi(h)=(h|_{\fn})_{\fn\in\Ps},\hspace{0.25cm}\mbox{for all~~}~~h\in\mathfrak{H}_K^\HQ.$$
{\em $\varphi$ is well defined:} For, let $h,k\in\mathfrak{H}_K^\HQ$ with $h=k$. Then $h,k:\HQ\longrightarrow \HQ$ are mappings and
$$\varphi(h)=(h|_{\fn})_{\fn\in\Ps}\hspace{1cm}\mbox{and}
\hspace{1cm}\varphi(k)=(k|_{\fn})_{\fn\in\Ps}.$$
Now $h=k$ implies that $h|_{\fn}=k|_{\fn},\hspace{0.25cm}\mbox{for all~~}~~\fn\in\Ps$. Thus
$$\varphi(h)=(h|_{\fn})_{\fn\in\Ps}=(k|_{\fn})_{\fn\in\Ps}=\varphi(k).$$
So the uniqueness property follows, and the closure property follows clearly.\\
{\em $\varphi$ is linear:} It follows trivially.\\
{\em $\varphi$ is injective:} For, let $h,k\in\mathfrak{H}_K^\HQ$ with $\varphi(h)=\varphi(k).$ Then
$$(h|_{\fn})_{\fn\in\Ps}=(k|_{\fn})_{\fn\in\Ps}$$
and so
$$h|_{\fn}=k|_{\fn},\hspace{0.25cm}\mbox{for all~~}~~\fn\in\Ps.$$
That is,
$$h(z)=k(z),\hspace{0.25cm}\mbox{for all~~}~~z\in \C_\fn\mbox{~~ and for all~~}~~\fn\in\Ps.$$
Since we have $\displaystyle \HQ=\bigcup_{\fn\in\Ps}\C_\fn$,
$$h(\qu)=k(\qu),\hspace{0.25cm}\mbox{for all~~}~~\qu\in H.$$
Thus $h=k$ and $\varphi$ is injective.\\
{\em $\varphi$ is surjective:} It follows clearly from the definition of $\varphi$.\\
Hence $\varphi$ is an isomorphism and the result follows.
\end{proof}
The Hilbert space $\mathfrak{H}_K^\HQ$ and $\overline{\mathfrak{H}}_K$ are in fact isometrically isomorphic. In the following proposition we validate this claim.
\begin{proposition}\label{p51}
The Hilbert space $\mathfrak{H}_K^\HQ$ is isometrically isomorphic to $\overline{\mathfrak{H}}_K$ up to a constant.
\end{proposition}
\begin{proof}
$\mathfrak{H}_K^\HQ$ is isomorphic to $\overline{\mathfrak{H}}_K$ has been done in the proposition (\ref{piso}). Now recall the isomorphism
$$\varphi:\mathfrak{H}_K^\HQ\longrightarrow\overline{\mathfrak{H}}_K$$
defined by $$\varphi(h)=(h|_{\fn})_{\fn\in\Ps},\hspace{0.25cm}\mbox{for all~~}~~h\in\mathfrak{H}_K^\HQ.$$
What is left to show is that $\varphi$ is an isometry. For, let $h,k\in\mathfrak{H}_K^\HQ$, then there exist $\{\alpha_{m}\},\{\beta_{m}\}\subset \HQ$ (note that they do not depend on $\fn\in \Ps$) such that
$$h=\sum_{m=0}^{\infty}\alpha_m\Phi_{m}\hspace{1cm}\mbox{and}\hspace{1cm}k
=\sum_{m=0}^{\infty}\beta_m\Phi_{m}$$
and so
$$~h|_{\fn}=\sum_{m=0}^{\infty}\alpha_mU_{m}\hspace{1cm}\mbox{and}\hspace{1cm}
k|_{\fn}=\sum_{m=0}^{\infty}\beta_mU_{m},$$
where $\Phi_m$ and $U_m$ are as in (\ref{ee6}) and (\ref{ee7}) respectively.
Now
\begin{eqnarray*}
\langle h|_{\fn}\vert k|_{\fn}\rangle_{\fn}&=&\left\langle \sum_{m=0}^{\infty}\alpha_{m}U_m\vert\sum_{n=0}^{\infty}\beta_{n}U_n\right\rangle_{\fn}\\
&=&\sum_{m=0}^{\infty}\sum_{n=0}^{\infty}\alpha_{m}\;\langle U_m\vert U_n\rangle_{\fn}\overline{\beta_{n}}\\
&=&\sum_{m=0}^{\infty}\alpha_{m}\overline{\beta_{m}}\hspace{1cm}\mbox{as~~}~~\langle U_m\vert U_n\rangle_{\fn}=\delta_{mn}.
\end{eqnarray*}
So
\begin{eqnarray*}
\langle \varphi(h)\vert \varphi(k)\rangle&=&\int_0^{\frac{\pi}{2}}\int_0^{2\pi}\langle h|_{\fn}\vert k|_{\fn}\rangle_{\fn}d\Omega(\theta, \phi)\\
&=&2\pi\sum_{m=0}^{\infty}\alpha_{m}\overline{\beta_{n}}\hspace{1cm}\mbox{as~~}~~\int_0^{\frac{\pi}{2}}
\int_0^{2\pi}d\Omega(\theta,\phi)=2\pi\\
&=&2\pi\sum_{m=0}^{\infty}\sum_{n=0}^{\infty}\alpha_{m}\;
\langle\Phi_m\vert\Phi_n\rangle~\overline{\beta_{n}}\hspace{1cm}\mbox{as~~}~
~\langle\Phi_m\vert\Phi_n\rangle=\delta_{mn}\\
&=&2\pi\left\langle\sum_{m=0}^{\infty}\alpha_{m}\Phi_m\vert\sum_{n=0}^{\infty}\beta_{n}\Phi_n\right\rangle\\
&=&2\pi\langle h\vert k\rangle.
\end{eqnarray*}
The conclusion follows.
\end{proof}
\begin{proposition}\label{p44}
$\overline{\mathfrak {H}}_K\subseteq\displaystyle \prod_{I\in\mathbb{S}}\mathfrak H_{K}^{\fn}$.
\end{proposition}
\begin{proof}
In view of Remark(\ref{RRem1}), it is straightforward.
\end{proof}
Define
$$\Psi_m~:~\Ps\longrightarrow\bigcup_{\fn\in\Ps}\mathfrak H_{K}^{\fn}\quad\text{by}\quad\Psi_m(\fn)=U_m,\quad\text{where}\quad U_m(z)
=\frac{\left(re^{\fn\theta}\right)^m}{\sqrt{2\pi x_m!}},$$
as in equation (\ref{ee6-b}).
Then clearly
$$\Psi_m(\fn)\in\mathfrak{H}_K^\fn;~~~~~~\forall m\in\mathbb{N},~~\forall \fn\in\Ps,$$
and thereby $\displaystyle \left\{\Psi_m\right\}_{m=1}^{\infty}\subseteq\prod_{\fn\in\Ps}\mathfrak{H}_K^\fn.$
\begin{proposition}\label{p45}
For each $m$ and $n$ the quaternion valued function $\fn\longrightarrow\langle\Psi_m(\fn)\vert\Psi_n(\fn)\rangle_{\fn}$ is $\nu$-measurable.
\end{proposition}
\begin{proof}
For each $m$ and $n$, we have
\begin{eqnarray*}
\langle\Psi_m(\fn)\vert\Psi_n(\fn)\rangle_{\fn}&=&\langle U_m\vert U_n\rangle_{\fn}\\
&=&\int_{\C_\fn}U_m(z)\overline{U_n(z)}d\nu(z,\oz)\\
&=&\int_{\C_\fn}\frac{(re^{\fn\theta})^m\overline{(re^{\fn\theta})}^n}
{\sqrt{2\pi x_m!}\sqrt{2\pi x_n!}}d\nu(z,\oz)\\
&=&\delta_{mn}.
\end{eqnarray*}
Therefore $\fn\longrightarrow\langle\Psi_m(\fn)\vert\Psi_n(\fn)\rangle_{\fn}$ is a constant function, thereby it is $\nu$-measurable.
\end{proof}

\begin{proposition}\label{p46}
For each $\fn\in\Ps$, the sequence of vectors $\displaystyle\left\{\Psi_n(\fn)\right\}_{n=0}^{\infty}$ is total in $\mathfrak{H}_K^\fn$
\end{proposition}
\begin{proof}
It follows from Equation (\ref{ee6-b}).
\end{proof}
The above two propositions assist Lemma (\ref{141}), thereby $\displaystyle\left\{\mathfrak{H}_K^\fn~~\vert~~\fn\in\Ps\right\}$ is a measurable field of Hilbert spaces and
$$\mathfrak {M}_K=\{\Psi\in\prod_{\fn\in \Ps}\mathfrak{H}_K^\fn~\mid~\fn\longmapsto\left\langle \Psi_n(\fn)\mid \Psi(\fn)\right\rangle_{\fn}~~\mbox{~~is~~}~\nu-\mbox{measurable for all~~}~n\}$$
is a $\nu$-measurable field of vectors.
\begin{proposition}\label{p47}$\displaystyle\overline{\mathfrak{H}}_K\subseteq\mathfrak{M}_K$.
\end{proposition}
\begin{proof}Let $\Psi\in\overline{\mathfrak{H}}_K$, then
$$\Psi=\left(h\vert_{\fn}\right)_{\fn\in\Ps}\quad\text{for some}\quad h\in\mathfrak{H}_K^\HQ.$$
Since $\quad h\in\mathfrak{H}_K^\HQ$, there exists $\{\alpha_n\}\subset \HQ$ (they do not depend on $\fn\in\Ps$) such that
$$h=\sum_{n=0}^{\infty}\alpha_n\Phi_n$$
and so for each $\fn\in\Ps$,
$$\Psi(\fn)=h\vert_{\fn}=\sum_{n=0}^{\infty}\alpha_nU_n.$$
Hence $\langle U_m\vert\Psi(\fn)\rangle_\fn=\alpha_m$ and $\alpha_m$ is a constant with respect to $\fn\in\Ps$.
Thereby, $\fn\longrightarrow \langle U_m\vert\Psi(\fn)\rangle_\fn$ is a constant function and so it is $\nu$-measurable.
 Therefore, $\Psi\in\mathfrak{M}_K$ and conclusion follows.
\end{proof}
\begin{proposition}\label{p48}
The maximal subspace of $\mathfrak{M}_K$ satisfying
\begin{equation}\label{m2}
\|\Phi\|^2:=\int_{\Ps}\|\Phi(\fn)\|_\fn^2d\Omega(\theta,\phi)<\infty
\end{equation}
is $\overline{\mathfrak{H}}_K$.
\end{proposition}
\begin{proof} Since $\mathfrak{H}_K^\fn\subseteq L_\HQ^2(\C_\fn,d\nu(z,\oz))$ for all $\fn\in\Ps$, we have
\begin{equation}\label{m3}
\|h\vert_\fn\|_\fn<\infty\quad\forall \fn\in\Ps~~~~~\text{and}~~~\forall h\in\mathfrak{H}_K^\HQ.
\end{equation}
Let $\Phi\in\overline{\mathfrak{H}}_K$ then $\displaystyle\Phi=(h\vert_{\fn})_{\fn\in\Ps}$ for some $h\in\mathfrak{H}_K^\HQ.$
That is, $\displaystyle\Phi(\fn)=h\vert_{\fn};\quad\forall \fn\in\Ps$. Thereby, $\|\Phi(\fn)\|_\fn^2<\infty$ for all $\fn\in\Ps$, and since $\Ps$ is compact, we have
$$\|\Phi\|^2:=\int_{\Ps}\|\Phi(\fn)\|_\fn^2d\Omega(\theta, \phi)<\infty.$$
Conversly, since  $d\Omega$ is a positive measure and $\Ps$ is a set of finite measure, if equation (\ref{m2}) holds, then $\|\Phi(\fn)\|_\fn^2<\infty$. Thus $\Phi(\fn)\in L_\HQ^2(\C_\fn, d\nu(z, \oz))$, in addition, if $\Phi\in\mathfrak{M}_K$, then $\Phi(\fn)$ must be a right-regular function. Therefore $\Phi(\fn)\in\mathfrak{H}_K^\fn$, for all $\fn\in\Ps$. Hence $\Phi\in\overline{\mathfrak{H}}_K$, which completes the proof.
\end{proof}
From the above arguments, now, one can write
\begin{equation}\label{m4}
\mathfrak{H}_K^\HQ\approxeq\overline{\mathfrak{H}}_K=\int_{\Ps}^{\oplus}\mathfrak{H}_K^\fn d\Omega(\theta, \phi).
\end{equation}

\section{Group representations}
In this section we shall obtain a set of irreducible unitary group representations on the representation space $V_{\fn}^L$, $\fn\in\Ps$ and show that their direct integral becomes a reducible representation of the bigger space $V_\HQ^L$.\\
For the standard complex harmonic oscillator canonical CS
\begin{equation}\label{G1}
|z\rangle =e^{-\frac{|z|^2}{2}}\sum_{m=0}^{\infty}\frac{z^m}{\sqrt{m!}}|n\rangle,
\end{equation}
we have the annihilation and creation operators as $a|n\rangle=\sqrt{n}|n-1\rangle$ and $a^{\dagger}|n\rangle=\sqrt{n+1}|n+1\rangle$ respectively. In this case we also have $a|z\rangle=z|z\rangle$. Further using the Baker-Campbell-Hausdorff identity, $\displaystyle e^{A+B}=e^{-\frac{1}{2}[A,B]}e^Ae^B$, when $A$ and $B$ commute with $[A,B]$, we have
$\displaystyle |z\rangle=e^{za^{\dagger}-\overline{z}a}|0\rangle$. Now by taking $\displaystyle z=\frac{q-ip}{\sqrt{2}}$ we can write
\begin{equation}\label{G2}
|z\rangle=e^{i(pQ-qP)}|0\rangle=U(q,p)|0\rangle,
\end{equation}
where $\displaystyle Q=\frac{a+a^{\dagger}}{\sqrt{2}},~~P=\frac{a-a^{\dagger}}{\sqrt{2}}$ and $U(q,p)$ is a unitary operator arising from a unitary, irreducible representation of the Weyl-Heisenberg group \cite{Alibk, Gaz}.
\subsection{Operators on left quaternion Hilbert spaces} Let $\mathcal{O}:V_\HQ^L\longrightarrow V_\HQ^L$ be a quaternion linear operator. In this case, the operators always act from the left as $\mathcal{O}|f\rangle$ and the scalar multiple of the operator is taken from the left as $\qu\mathcal{O}$. Note that the quaternion scalar multiples of an operator do not obey several properties of their complex counterpart \cite{Mu}.  Further the operators obey the following rules:
\begin{enumerate}
\item[(i)]$\displaystyle\mathcal{O}|\qu f\rangle=\qu(\mathcal{O}|f\rangle).$
\item[(ii)]$\displaystyle(\qu\mathcal{O})|f\rangle=\mathcal{O}|\qu f\rangle.$
\end{enumerate}
For a detail explanation we refer the reader to \cite{Ad}. Note that for $\qu\in \HQ$, in general $(\qu\mathcal{O})^{\dagger}\not=\oqu\mathcal{O}^{\dagger}$. However, for $\qu\in\Cn$ and $\mathcal{O}:V_\fn^L\longrightarrow V_\fn^L$ we have $(\qu\mathcal{O})^{\dagger}=\oqu\mathcal{O}^{\dagger}$. For details see \cite{Mu}.
\subsection{Group representation on $V_\HQ^L$} Let $\displaystyle\{f_m\}_{m=0}^{\infty}$ be an orthonormal basis of $V_\HQ^L$. Define the annihilation, creation and number operators,
$$a_L,a_L^{\dagger}, N_L:V_\HQ^L\longrightarrow V_\HQ^L,$$ as usual by
\begin{eqnarray*}
a_Lf_m&=&\sqrt{m}f_{m-1};~~~a_Lf_0=0,\\
a_L^{\dagger}f_m&=&\sqrt{m+1}f_{m+1},\\
N_Lf_m&=&mf_m.
\end{eqnarray*}
One can easily see that $a_L^{\dagger}$ is the adjoint of $a_L$ and $N_L=a_L^{\dagger}a_L$. Further, for the states in (\ref{quat-cs1}), in \cite{Thi}, we showed that $\displaystyle a_L\eta_{\qu}=\qu\eta_{\qu}$ and
\begin{equation}\label{G3}
\eta_{\qu}=e^{\qu a_L^{\dagger}-\overline{\qu}a_L}f_0.
\end{equation}
However, we also proved in \cite{Thi} that the operator in (\ref{G3}) cannot be identified as a group representation of the representation space $V_\HQ^L$ (see also \cite{AdMil} for the same result in a different view).
\subsection {Group representation on $V_{\fn}^L$} Let $\displaystyle\{g_m\}_{m=0}^{\infty}$ be an orthonormal basis of $V_{\fn}^L$. Define the left-annihilation, left-creation and left-number operators,
$$a_\fn,a_\fn^{\dagger}, N_\fn:V_{\fn}^L\longrightarrow V_{\fn}^L,$$ as usual by
\begin{eqnarray*}
a_\fn g_m&=&\sqrt{m}g_{m-1};~~~a_\fn f_0=0,\\
a_\fn^{\dagger}g_m&=&\sqrt{m+1}g_{m+1},\\
N_\fn g_m&=&mg_m.
\end{eqnarray*}
One can easily see that $a_\fn^{\dagger}$ is the adjoint of $a_\fn$ and $N_\fn=a_\fn^{\dagger}a_\fn$. Further, for $\qu\in \C_\fn$, for the states
\begin{equation}\label{G4}
\eta_{\qu}=e^{-\frac{|\qu|^2}{2}}\sum_{m=0}^{\infty}\frac{\qu^m}{\sqrt{m!}}g_m\in V_{\fn}^L,
\end{equation}
as in \cite{Thi}, we can show that $\displaystyle a_\fn\gamma_{\qu}=\qu\gamma_{\qu}$ and
\begin{equation}\label{G5}
\gamma_{\qu}=e^{\qu a_\fn^{\dagger}-\overline{\qu}a_\fn}g_0.
\end{equation}
Here, not as in the case of the Hilbert space $V_\HQ^L$, we shall show that the operator in (\ref{G5}) is a unitary operator and it arises from a unitary irreducible representation of the Weyl-Heisenberg group.\\

The operators $a_\fn, a_\fn^{\dagger}$ and $N_\fn$ satisfy the usual commutation relations,
\begin{equation}\label{G6}
[a_\fn,a_\fn^{\dagger}]=\mathbb{I}_\fn,\quad[N_\fn,a_\fn]=-a_\fn,\quad [N_\fn,a_\fn^{\dagger}]=a_\fn^{\dagger},
\end{equation}
where $\mathbb{I}_\fn$ is the identity operator on $V_{\fn}^L$.
Since the elements of $\Cn$ commute, unlike in the case of $\HQ$ (see \cite{Thi}), the algebra,
\begin{equation}\label{G7}
\mathcal{A}_{WH}=\text{linear span over} \Cn\left\{a_\fn,a_\fn^{\dagger}, \mathbb{I}_\fn\right\}
\end{equation}
is closed under $\Cn$ and is a version of the Weyl-Heisenberg algebra. Further, the operator $X=\fn y\mathbb{I}_\fn+(\qu a^{\dagger}_\fn-\overline{\qu}a_\fn)$ is anti-self adjoint in $V_{\fn}^L$ and is the infinitessimal generator of the operator
\begin{equation}\label{G8}
e^X=e^{\fn y\mathbb{I}_\fn}e^{\qu a^{\dagger}_\fn-\overline{\qu}a_\fn}=e^{\fn y\mathbb{I}_\fn}D(\qu).
\end{equation}

\begin{proposition}\label{PG1}
For $\qu\in \Cn$, the operator $D(\qu)=\displaystyle e^{\qu a_\fn^{\dagger}-\overline{\qu}a_\fn}$ is unitary.
\end{proposition}
\begin{proof}
Let $A=\qu a_\fn^{\dagger}-\overline{\qu}a_\fn$, then $A^{\dagger}=\overline{\qu} a_\fn-\qu a_\fn^{\dagger}=-A$ and $\displaystyle (e^A)^{\dagger}=e^{A^{\dagger}}$. Therefore, from the Baker-Campbell-Hausdorff formula we have
$$D(\qu)D(\qu)^{\dagger}=e^A(e^A)^{\dagger}=e^Ae^{-A}=e^{\frac{1}{2}[A,-A]}e^{A+(-A)}=\mathbb{I}_\fn.$$
\end{proof}
\begin{proposition}\label{PG-1}
For $\qu\in \Cn$, the operator $D(\qu)=\displaystyle e^{\qu a_\fn^{\dagger}-\overline{\qu}a_\fn}$ is a unitary representation of the representation space $V_{\fn}^L$.
\end{proposition}
\begin{proof}
Let $\qu_1=x_1+\fn y_1,~\qu_2=x_2+\fn y_2\in \Cn$ and $A_i=\qu_i a_\fn^{\dagger}-\overline{\qu}_ia_\fn;~~i=1,2.$
Since $\qu_1$ and $\qu_2$ commute, it can be easily checked that
$$[A_1,A_2]=-2\fn(x_1y_2-x_2y_1),$$
thereby $[A_1,A_2]$ commute with $A_1$ and $A_2$. Applying the Baker-Campbell-Hausdorff formula we get
\begin{eqnarray*}
D(\qu_1)D(\qu_2)&=&e^{A_1}e^{A_2}\\
&=&e^{\frac{1}{2}[A_1,A_2]}e^{A_1+A_2}\\
&=&e^{-\fn(x_1y_2-x_2y_1)}e^{(\qu_1+\qu_2)a^{\dagger}_\fn-(\overline{\qu}_1+\overline{\qu_2})a_{\fn}\fn}\\
&=&e^{-\fn(x_1y_2-x_2y_1)}e^{q_3a^{\dagger}_\fn-\overline{\qu}_3a_\fn};\quad\text{with}~~\qu_3=\qu_1+\qu_2\\
&=&e^{-\fn(x_1y_2-x_2y_1)}D(\qu_3)\\
&=&e^{-\fn\qu_1\wedge\qu_2}D(\qu_3).
\end{eqnarray*}
Therefore $D(\qu)$ is a unitary representation up to a phase factor.
\end{proof}

The irreducibility follows as in the complex case. Hence, $D(\qu)$ is a unitary irreducible representation of the representation space $V_{\fn}^L$.\\

In the following we shall employ a  transformation to transfer the operator, $D(\qu)$ and CS, $\gamma_{\qu}$ from $V_{\fn}^L$ to $V_{\fn'}^L$. In order to enhance the distinquishability let us rewrite the CS $\gamma_{\qu}$ as follows:
\begin{equation}\label{G10}
\gamma_{\qu}^\fn=e^{-\frac{r^2}{2}}\sum_{m=0}^{\infty}\frac{\left(re^{\fn\theta}\right)^m}{\sqrt{m!}}g_m^\fn
=D\left(re^{\fn\theta}\right)g_0^\fn\in V_{\fn}^L,
\end{equation}
where $\displaystyle\{g_m^\fn\}_{m=0}^{\infty}$ is an orthonormal basis of $V_{\fn}^L$.
Let $\mathcal{L}(V_\HQ^L)$ be the set of all linear operators on $V_\HQ^L$. Define the operator valued function,
\begin{equation}\label{G11}
F:\Ps\longrightarrow \mathcal{L}(V_\HQ^L)\quad\text{by}~~~F(\fn')=\vert\gamma_{\qu}^{\fn'}\rangle\langle\gamma_{\qu}^\fn\vert,
\end{equation}
where one should be clear with the notion that the $\qu$ in $\gamma_{\qu}^{\fn'}$ is in $\C_{\fn'}$ and the $\qu$ in $\gamma_{\qu}^\fn$ is in $\Cn.$ With this transformation, it is straightforward that
\begin{equation}\label{G12}
F(\fn')\gamma_{\qu}^\fn=\gamma_{\qu}^{\fn'}=D(re^{\fn'\theta})g_0^{\fn'}.
\end{equation}
Define a new operator $\mathfrak{D}(\qu):V_\HQ^L\longrightarrow V_\HQ^L$ as follows:
\begin{equation}\label{G13}
\mathfrak{D}(\qu)=\int_{\Ps}^{\oplus}D(re^{\fn\theta})d\nu(\fn).
\end{equation}
The decomposable operator $\mathfrak{D}(\qu)$ is a reducible representation of the representation space $V_\HQ^L$ and it acts on the CS, $\eta_{\qu}$ as
$$\mathfrak{D}(\qu)\eta_{\qu}=D(\qu)f_0\vert_{\fn}\quad\text{for}~~\qu\in \Cn;~~\fn\in\Ps.$$
It should be noted that the operator $\mathfrak{D}(\qu)$ is different from the operator in equation (\ref{G5}).

\section{Conclusion}
Using the CS constructed on a set of quaternionic Hilbert spaces we have obtained a measurable field of reproducing kernel Hilbert spaces and their direct integral. However, there is a converse approach, whenever we are given a measurable family of Hilbert spaces and a positive definite-kernel on it one can obtain a reproducing kernel and an associated Hilbert space and thereby a set of CS \cite{Alibk}. Further, the theory of positive operator valued functions (POVs) and positive operator valued measures (POVMs) are allied with a measurable family of Hilbert spaces \cite{HP, Alibk}. In fact, in complex quantum mechanics normalized POVMs are identified with quantum observables. In \cite{HP}, starting with a measurable field of Hilbert spaces the authors obtained extremal POVMs and in physical point of view extremal observables describe quantum measurements that are free from any classical randomness. In this regard, in the quaternion quantum mechanics, the converse approach and the role of POVs and POVMs, along the lines of \cite{HP, Alibk}, are yet to be seen.
\section{Appendix}
\subsection{Coherent states: General construction}
The scheme of the general construction is borrowed from \cite{ThiAli}. Let $X$ be any locally compact space with a (Radon) measure $\nu$ on it and $\Phi_m : X \longrightarrow H , \;\; m \in\mathbb{N} ,$ be a sequence of functions which satisfy the two conditions,
\begin{itemize}
\item[1.] $\displaystyle 0 < \mathcal N (x) := \sum_{m=0}^{\infty}\vert\Phi_m (x)\vert^2 < \infty$, for all $x \in X$.\\
\item[2.] $\displaystyle\int_X \Phi_m (x)\overline{\Phi_n (x)}\; d\nu (x) = \delta_{m n}$, for all $m$ and $n$.
\end{itemize}
\noindent
A family of coherent states
$\{\eta_{x}~\mid~x\in X\}\subseteq V_{H}^{L}$ can be defined to be the vectors,
\begin{equation}\label{eq9}
\eta_x = \mathcal N (x)^{- \frac{1}{2}} \sum_m \Phi_m (x) \phi_m\; ;
\end{equation}
where $\{\phi_{m}\}_{m=0}^{\infty}$ is an orthonormal basis of $V_{H}^{L}$.
Now by construction
$$\Vert \eta_x \Vert^2 =1, \mbox{~~for all~~} x\in X,$$
and satisfy the resolution of identity,
$$ \int_X\langle f \mid \eta_x \rangle \langle \eta_x \mid g \rangle \; \mathcal N (x)d\mu (x)\;
= \langle f \mid g\rangle \; , \qquad f,g \in V^L_H \; .$$
Moreover, taking $L^2_H (X, d\nu )$ to be a left quaternionic Hilbert space, the map
\begin{equation}
W: V^L_H \longrightarrow L^2_H (X, d\nu ), \quad \text{with} \quad  W f (x) = \mathcal N (x)^{\frac 12}\langle f\mid \eta_x\rangle_{V^L_H}
\label{CS-H_isom}
\end{equation}
 is a linear isometry onto a closed subspace
 $$ \mathfrak H_K^X := WV^L_H \subset  L^2_H (X,  d\nu ).$$
Moreover, the space $\mathfrak{H}_K^X$ is a reproducing kernel Hilbert space, with reproducing kernel
\begin{equation}\label{ee1}
K:X\times X\longrightarrow H,\quad K(y,x)=\left[\mathcal{N}(y)\mathcal{N}(x)\right]^{1/2}\langle\eta_y\vert\eta_x\rangle=\sum_{m}\Phi_m(y)\overline{\Phi_m(x)}.
\end{equation}

\end{document}